\documentclass[11pt,letterpaper]{article}
\usepackage[top=1in, bottom=1in, left=1in,right=1in]{geometry}
\usepackage{amsmath,amssymb}
\usepackage{amsthm}
\usepackage{hyperref}

\begin{document}

\newtheorem{theorem}{Theorem}
\newtheorem{lemma}{Lemma}
\newtheorem{proposition}{Proposition}
\newtheorem{corollary}{Corollary}
\newtheorem{definition}{Definition}
\newcommand{\ket}[1]{|#1\rangle}
\newcommand{\bra}[1]{\langle #1|}
\newtheorem{fact}{Fact}
\newcommand{\Z}{\mathbb Z} 
\newcommand{\F}{\mathbb F}
\newcommand{\N}{\mathbb N}
\newcommand{\C}{\mathbb C}
\newcommand{\Fqd}{\F_q^{(d)}[x]}
\newcommand{\HSP}{\mbox{\rmfamily\textsc{HSP}}}
\newcommand{\HPP}{\mbox{\rmfamily\textsc{HPP}}}
\newcommand{\HPFGP}{\mbox{\rmfamily\textsc{HPGP}}}
\newcommand{\HQPP}{\mbox{\rmfamily\textsc{HQPP}}}
\newcommand{\HSOP}{\mbox{\rmfamily\textsc{HSSP}}}
\newcommand{\HTP}{\mbox{\rmfamily\textsc{Hidden Translation}}}
\newcommand{\OP}{\mbox{\rmfamily\textsc{Orbit Coset}}}
\newcommand{\Aff}{\mbox{\rm Aff}}
\newcommand{\Fg}{\mbox{\rm Fg}}
\newcommand{\Affq}{{\rm Aff}_q}
\newcommand{\Affp}{{\rm Aff}_p}
\newcommand{\Dihq}{{\rm Dih}({\mathbb F}_q)}
\newcommand{\GL}{\rm GL}
\newcommand{\Dih}{\mbox{\rm Dih}}
\title{Hidden Symmetry Subgroup Problems } 

\date{January 18, 2012}

\author{Thomas Decker\\ 
Centre for Quantum Technologies, National University
    of Singapore, \\ Singapore 117543.
    \texttt{cqttd@nus.edu.sg}
\and
G{\'a}bor Ivanyos\\
Computer and Automation Research Institute \\ of the
    Hungarian Academy of Sciences,  \\ Budapest,
    Hungary. \texttt{Gabor.Ivanyos@sztaki.hu}
\and
Miklos Santha\\
LIAFA, Univ. Paris 7, CNRS,
75205 Paris, France;  and \\
Centre for Quantum Technologies,
National University of Singapore, \\ Singapore 117543. \texttt{miklos.santha@liafa.jussieu.fr}
\and
Pawel Wocjan\\
Department of Electrical
    Engineering and Computer Science, \\ University of Central Florida,
    \\ Orlando, FL 32816-2362, USA. 
    \texttt{wocjan@eecs.ucf.edu}
}

\maketitle

\begin{abstract}
We advocate a new approach of addressing hidden structure problems and
finding efficient quantum algorithms.  We introduce and investigate
the Hidden Symmetry Subgroup Problem (\HSOP{}), which is a
generalization of the well-studied Hidden Subgroup Problem (\HSP{}).
Given a group acting on a set and an oracle whose level sets define a
partition of the set, the task is to recover the subgroup of
symmetries of this partition inside the group.  The $\HSOP$ provides a
unifying framework that, besides the $\HSP$, encompasses a wide range
of algebraic oracle problems, including quadratic hidden polynomial
problems.  While the \HSOP{} can have provably exponential
quantum query complexity, we obtain efficient quantum algorithms for
various interesting cases.  To achieve this, we present a
general method for reducing the \HSOP{} to the \HSP{}, which works
efficiently in several cases related to symmetries of polynomials.
The \HSOP{} therefore connects in a rather surprising way certain
hidden polynomial problems with the \HSP{}.  Using this connection, we
obtain the first efficient quantum algorithm for the hidden polynomial
problem for multivariate quadratic polynomials over fields of constant
characteristic. We also apply the new methods to polynomial function
graph problems and present an efficient quantum procedure
for constant degree multivariate polynomials over any field.
This result improves in several ways 
the currently known
algorithms.
\end{abstract}

\section{Introduction}
The main goal of quantum computing is to identify suitable classes of
problems and to find efficient quantum algorithms for them that
provide a significant speed-up over their classical counterparts. The
vast majority of such examples consists of group-theoretical problems
that can be formulated within the framework of the hidden subgroup
problem ($\HSP$).  This problem can be cast in the following terms: We
are given a finite group $G$ and a black-box function from $G$ to some
finite set. The level sets of the function correspond to the right
cosets of some subgroup $H$. We say that $f$ hides $H$ and the task is
to determine this hidden subgroup. One query of the function counts as
one step in the computation and an algorithm is efficient if its
running time is polynomial in the logarithm of the size of the group.
While no classical algorithm is known to solve this problem with
polynomial query complexity, the problem is computationally solvable
in quantum polynomial time for every abelian
group~\cite{sho97,BoLi,Kitaev}.

Several attempts were made to extend the quantum solution of the
abelian HSP. Most of the research focused on the HSP in non-abelian
groups since these include several algorithmically important problems.
For example, it is known that efficient solutions for the dihedral and
the symmetric group would imply efficient solutions for some lattice
problems~\cite{Regev} and for graph isomorphism, respectively.  While
some progress has been made in this
direction~\cite{BCvD05,DMR10,FIMSS03,gsvv01,hrt03,Kuperberg,MRRS04},
the HSP for the dihedral and symmetric groups remains unsolved.  It is
already known that the methods for solving the abelian case fail for
several non-abelian groups~\cite{MRS,hmrrs}.  The goal of obtaining
efficient quantum algorithms for larger classes of non-abelian groups
turned out to be rather elusive.

Another idea for generalizing the problem was proposed by Childs,
Schulman and Vazirani~\cite{CSV} who considered properties of
algebraic sets hidden by black-box functions. One of these problems is
the hidden polynomial problem ($\HPP$) where the hidden object is a
polynomial. To recover it we have at our disposal an oracle whose
level sets coincide with the level sets of the polynomial. Childs et
al.~\cite{CSV} showed that the quantum query complexity of this
problem is polynomial in the logarithm of the field size provided that
the degree and the number of variables are held constant, leaving the
question of the time complexity as an open question.  The authors also
formulated computationally efficient quantum procedures for some
related problems, such as the hidden radius and the hidden flat of
centers. Nonetheless, to the best of our knowledge, no efficient
quantum polynomial time algorithm has been proposed for the general
$\HPP$, not even for the simplest problem of hidden quadratic
polynomials in one variable $(\HQPP)$.

In~\cite{DDW09}, Decker, Draisma and Wocjan defined a related problem
that we refer to as the hidden polynomial graph problem (\HPFGP) to
distinguish it from the $\HPP$.  Here, similarly to the $\HPP$, the
hidden object is a polynomial, but the oracle is more powerful because
it can also be queried on the graphs that are defined by the
polynomial functions.  They obtained a polynomial time quantum
algorithm that correctly identifies the hidden polynomial when the
degree and the number of variables are considered to be constant.
Their proof applies to all finite fields whose characteristic is not
in a finite set of exceptional characteristics that depend on the
degree of the polynomials.

In this paper, we advocate a third possible approach to find hidden
structures. We consider a group $G$ acting on some finite set $M$, and
we suppose that we have at our disposal a black-box function whose
level sets define a partition of $M$. The object we would like to
recover is the group of symmetries of this partition inside $G$, i.e.,
the subgroup whose orbits under the action coincide with the classes
of the partition. We call this problem the hidden symmetry subgroup
problem ($\HSOP$). It is easy to see that the $\HSP$ is a special case
of the $\HSOP$ when the group acts on itself and the action
corresponds to the group operation. But, for some actions, the $\HSOP$
is provably harder than any $\HSP$.  We show that Grover's search can
be cast as an $\HSOP$, establishing that certain cases of the
$\HSOP$ have exponential quantum query complexity.  This is in
contrast to the $\HSP$ that has polynomial quantum query complexity
for all groups~\cite{EHK04}.

The potential of the $\HSOP$ lies mainly in the possibility of
extending the HSP techniques to more general group actions that still
admit efficient quantum procedures.  We demonstrate the power of this
new approach by designing and improving quantum algorithms for several
algebraic problems.  To achieve this we reduce both the $\HQPP$ and
the univariate $\HPFGP$ to appropriate $\HSOP$s for
which we can give efficient quantum solutions in some interesting
cases.  Besides the construction of efficient algorithms, the
formulation of problems as $\HSOP$ can also shed new light on their
structure. For example, the apparent difficulty of the $\HQPP$ over
prime fields might be explained by the equivalence of this problem to
the $\HSP$ in the dihedral group, a connection discovered via their
relations to the $\HSOP$. It is also worth to note that the hidden
shifted multiplicative character problem of van Dam, Hallgren and
Ip~\cite{vDHI03} is a version of the $\HSOP$ with an additional
promise on the input.

To establish our algorithmic results, we first concentrate on the
question of whether the $\HSOP$ can be reduced in some cases to the
related $\HSP$ that we obtain by forgetting about the action. We
design a reduction scheme, which involves the generalization of bases
known from the theory of permutation groups. We are able to show that
when the action has an efficiently computable generalized base then
the $\HSOP$ is indeed efficiently reducible to the related $\HSP$
({\bf Proposition~\ref{prop:reduction}}).  Then we describe a
probabilistic construction of such bases for a large class of
Frobenius groups. Therefore, the above reduction applies to these
groups ({\bf Theorem~\ref{theorem:frobenius-reduction}}).  These
groups include among others a large variety of affine groups and the
$\HSOP$ is efficiently solvable for these groups by a quantum
algorithm. We remark that in \cite{MRRS07} it is proved that the
$\HSOP$ (in a slightly different formulation) can be solved
efficiently for some of these affine groups. The proof uses
essentially the same reduction technique.

We then establish several surprising connections between hidden
polynomial problems and the $\HSOP$.  In fact, the $\HQPP$ turns out
to be equivalent in a very strong sense to the $\HSOP$ over a related
affine group.  Combined with the above reduction to the related
$\HSP$, we are able to give the first ever quantum polynomial time solution for the
$\HQPP$ over fields of constant characteristic ({\bf
  Theorem~\ref{theorem:hqp-small}}).  We then give a quantum reduction
of the multivariate quadratic $\HPP$ to the $\HQPP$, which implies
that over fields of constant characteristic this multivariate problem
is also solvable in quantum polynomial time ({\bf
  Theorem~\ref{theorem:multivariate}}).

Finally, for dealing with the $\HPFGP$, we define a class of
semidirect product groups which we call function graph groups. We show
that the $\HPFGP$ for univariate polynomials of degree at most $d$
coincides with the $\HSOP$ over a corresponding function graph
group. These groups turn out to have a base of size $d$, and therefore
our general reduction to the related $\HSP$ applies ({\bf
  Theorem~\ref{theorem:hpgptohsp}}).  Based on this reduction, we
improve the results of~\cite{DDW09} by showing that there is a quantum
polynomial time algorithm for the $\HPFGP$ over every field 
when the degree of the polynomials is
constant ({\bf Theorem~\ref{theorem:new-reduction}}).

\section{Preliminaries}
We first fix some useful notation: $n$ denotes a positive integer, $p$
a prime number, $q$ a prime power, $\Z_n$ the additive group of
integers modulo $n$, $\F_q$ the finite field of size $q$, and $\Fqd$
the set of univariate polynomials of degree at most $d$ over $\F_q$.

\subsection{Level sets and problem classes}
Simply speaking, we study the general problem of determining hidden
objects related to a given algebraic structure.  The algebraic
structure is specified by parameters of the problem, which are finite
groups, families of subgroups of a given group, group actions, finite
fields, and integers in the present case.  We assume that we have
access to an unknown member of a family of black-box functions $f : A
\to S$, where $A$ is part of the structure and $S$ is some finite set.
We consider this function $f$ as the oracle input.  We are restricted
to identifying the hidden object solely from the information we obtain
by querying the oracle $f$. In fact, the only useful information we
can obtain is the structure of the {\em level sets} $f^{-1}(s)=\{ a
\in A: f(a)=s\}$, $s\in S$, that is, we can only determine whether two
elements in $A$ are mapped to the same value or not. All non-empty
level sets together constitute a partition of $A$ which we denote by
$\pi_f$.

\begin{definition}\label{def 1}
{\rm
The {\em hidden subgroup problem} $\HSP$ is parametrized by a finite group $G$ and
a family ${\cal H}$ of subgroups of $G$.\\
\vbox{\begin{quote}
\HSP$(G, {\cal H})$ \\
{\em Oracle input:} 
A function $f$ from $G$ to some finite set $S$ such that for some subgroup $H \in {\cal H}$, we have
$f(x) = f(y) \Longleftrightarrow Hx = Hy.$\\
{\em Output:} $H$.
\end{quote}}\\
The {\em hidden polynomial problem} $\HPP$ is parametrized by a finite field $\F_q$ and two positive integers $n$ and $d$.\\
\vbox{\begin{quote}
\HPP$(\F_q, n,d)$.\\
{\em Oracle input:} 
A function $f$ from $\F_q^n$ to some finite set $S$ 
such that for some $n$-variate polynomial ${\cal P}$ of degree $d$ over $\F_q$, we have
$f(x) = f(y) \Longleftrightarrow {\cal P}(x) = {\cal P}(y).$\\
{\em Output:} ${\cal P}$.
\end{quote}}\\
For every $u \in \F_q$ we define a monic quadratic polynomial over $\F_q$ by 
${\cal P}_u(x) = x^2 -2ux$.
The {\em hidden quadratic polynomial problem} $\HQPP$ is parametrized by some finite field ${\F_q}$.\\
\vbox{\begin{quote}
\HQPP$(\F_q)$.\\
{\em Oracle input:} 
A function $f$ from $\F_q$ to some finite set $S$
such that we have 
$f(x) = f(y) \Longleftrightarrow {\cal P}_u(x) = {\cal P}_u(y).$\\
{\em Output:} ${\cal P}_u$.
\end{quote}}\\
The {\em hidden polynomial graph problem} $\HPFGP$ is parametrized 
by a finite field $\F_q$ and two positive integers $n$ and $d$.\\
\vbox{\begin{quote}
\HPFGP$(\F_q, n,d)$.\\
{\em Oracle input:} 
A function $f$ from $\F_q^n \times \F_q$ to a finite set $S$ 
such that for some $n$-variate polynomial $Q$ of degree $d$ over $\F_q$ we have
$f(x_1,y_1) = f(x_2, y_2) \Longleftrightarrow y_1 - Q(x_1) = y_2 - Q(x_2).$\\
{\em Output:} $Q$.
\end{quote}}\\
In all these problems we say that the input $f$ {\em hides} the output
of the problem.  }
\end{definition}

In the definition of the HQPP we restrict our attention to monic
polynomials with zero constant term because adding a constant to a
polynomial or multiplying all coefficients with the same non-zero
constant do not change the partition $\pi_f$.  Furthermore, observe
that for the $\HPFGP$ we have a more powerful oracle at our disposal
than for the $\HPP$, because an $\HPFGP$ oracle $f(x,y)$
restricted to $y = 0$ is equivalent to an $\HPP$ oracle.

In all these problems the task is to determine the output hidden by
the oracle input.  We measure the time complexity of an algorithm by
the overall running time when a query counts as one computational
step.  An algorithm is {\em efficient} if its time complexity is
polynomial in the logarithm of the size of the group or field, and in
the size of the integers in unary in the parametrization of the
problem. 

\subsection{Semidirect product groups}
Let $K$ and $H$ be finite groups and let $\phi:h\mapsto \phi_h$ be
a homomorphisms from $H$ to the group of automorphisms
of $K$. Then the semidirect product $K\rtimes_\phi H$
is the cartesian product of $K$ and $H$ equipped
with the multiplication defined as
$(k,h)\cdot (k',h')=(k\cdot\phi_h(k'),h\cdot h')$.
We use the notation $K\rtimes H$ for $K\rtimes_\phi H$ whenever $\phi$
is clear from the context.

\subsection{Group actions and partitions}

A {\em left permutation action} of a group $G$ on a set $M$ is a
binary function $ \circ: G \times M \to M\,, $ where we denote $\circ
(g,m)$ by $g \circ m$, which for all $g,h \in G$ and $m \in M$
satisfies $ g \circ (h \circ m) = (gh) \circ m $ {\rm and} $ e \circ m
= m $ for the identity element $e$ of $G$. For a subset $L \subseteq
M$ we set $g \circ L = \{g \circ m : m \in L\}.$ The {\em stabilizer}
subgroup $G_m$ of $m$ is defined as $\{g \in G : g \circ m = m\},$ it
consists of the elements in $G$ which fix $m$.  The action $\circ$ is
{\em faithful} if $\bigcap_{m \in M} G_m = \{e\}$.  Throughout the
paper we assume faithfulness.  If $G$ acts on $M$, then every subgroup
$H$ of $G$ acts also naturally on $M$.  The $H$-{\em orbit} of $m \in
M$ is the subset of $M$ where $m$ can be moved to by elements of $H$,
formally $H \circ m = \{h \circ m : h\in H\}.$

The $H$-orbits form a partition $H^* =\{ H \circ m : m \in M \}$  of $M$.
For a partition $\pi=\{\pi_1, \ldots,
\pi_\ell\}$ of the set $M$, we define the subgroup
$
\pi^*   =\{ g \in G : (\forall i) \; g \circ \pi_i = \pi_i \}.
$
We call $\pi^* \leq G$ the group of {\em symmetries} of $\pi$ within $G$.
This is the subgroup of elements that stabilize every class of the
partition $\pi$ under the given action.  Let $({\cal S}(G),
\subseteq)$ be the lattice of subgroups of $G$ under the inclusion
relation, and let $(\Pi(M), \leq)$ be the lattice of partitions of
$M$, where by definition $\pi \leq \pi^\prime$ if $\pi'$ is finer than
$\pi$.  The maps $H \mapsto H^*$ and $\pi \mapsto \pi^*$ define an
order-reversing Galois connection between $({\cal S}(G), \subseteq)$
and $(\Pi(M), \leq)$, that is
$H \leq \pi^*$ if and only if $\pi \subseteq H^*$.
The subgroup $H^{**}$  is the {\em closure} of $H$~\cite{Blyth}, it consists of 
the elements in $G$ which stabilize every $H$-orbit. The {\em closure} of a partition $\pi$ is
$\pi^{**}$, it consists of the orbits of its group of symmetries. 
It is always true that $H \subseteq H^{**}$.
The subgroup $H$
 is {\em closed} if $H = H^{**}$,
or equivalently, there exists a partition $\pi$ such that $H = \pi^{*}$. Similarly, 
$\pi$ is {\em closed} if $\pi = \pi^{**}$.
We denote by ${\cal C}(G)$ the family of all closed subgroups in $G$.

\subsection{The hidden symmetry subgroup problem}

\begin{definition}
{\rm
The {\em hidden symmetry subgroup problem} $\HSOP$ is parametrized by a finite group $G$,
a finite set $M$, an action $\circ: G \times M \to M$  of $G$ on $M$, and
a family ${\cal H}$ of closed subgroups of $G$.\\
\vbox{\begin{quote}
\HSOP$(G, M, \circ, {\cal H})$.\\
{\it Oracle input:} 
A function $f$ from $M$ to some finite set $S$ such that for some subgroup $H \in {\cal H}$, we have
$f(x) = f(y) \Longleftrightarrow H \circ x = H \circ y.$\\
{\it Output:} $H$.
\end{quote}}
}
\end{definition}

In general, there can be several subgroups whose orbits coincide with
the level sets of $f$, but the closures of these subgroups are the
same. The unique closed subgroup that satisfies the promise is
$\pi_f^*$, and this is exactly the output of the problem. We will say
that $f$ {\em hides $H$ by symmetries}.  In fact, it would be natural
to extend \HSOP{} to the more general setting where $f$ is an
arbitrary function on $M$ and the task is to determine the (closed)
subgroup $\pi_f^*$. The restriction we use in this paper is that
$\pi_f$ is a closed partition with $\pi_f^*\in {\cal H}$.
We define an algorithm for solving the HSSP as {\em efficient} if 
it is polylogarithmic in $|G|$.

It is easy to see that the HSP is a special case of the HSSP when we
set $M=G$ and choose the group action $\circ$ to be the group
operation, that is $g \circ h = gh$.  For this action every subgroup
of $G$ is closed, and a function $f$ hides a subgroup $H$ if and only
if $f$ hides $H$ by symmetries.

Given \HSOP$(G, M, \circ, {\cal H})$, by forgetting about the group
action we obtain \HSP$(G, {\cal H})$.  We call this problem the {\em
  related} $\HSP$.

\subsection{Related results}
While the $\HSP$ is generally hard in non-abelian groups, its query
complexity is always small, due to a classical result of Ettinger,
H{\o}yer and Knill~\cite{EHK04}.
\begin{fact}\label{fact:ettinger}
For every finite $G$, the $\HSP(G, {\cal C}(G))$ has polynomial query complexity.
\end{fact}
Among groups where the $\HSP$ is solvable in quantum polynomial time, some affine groups will be of importance for us.
For a subgroup $H$ of $\F_q^*$, let $\Affq(H)$ denote the semidirect product 
$\F_q \rtimes H$, and let ${\cal FC}$ be the family of 
conjugates of $H$ by an element of $\F_q$ 
(for a detailed discussion
of these groups see Section~\ref{section:affine}). The following positive results on the solvability of the $\HSP$ were obtained 
respectively by Moore et al.~\cite{MRRS04} and Friedl et al.~\cite{FIMSS03}.
\begin{fact}\label{fact:friedl}
The following cases of the $\HSP$
can be solved in polynomial time:
\begin{description}
\item[(a)]
$\HSP(\Affq(H), {\cal FC})$, where
$q$ is a prime and $H \leq \F_q^*$ such that $1 < |H| \leq q-1$ 
and $|H| = \Omega (q/ {\rm polylog}(q))$.
\item[(b)]
$\HSP(G,{\cal  C}(G))$, where 
$G$ is a finite group such that $G'$ is commutative
and every element of $G'$ has an order bounded by a constant.
\end{description}
\end{fact}
The query complexity of the \HPP{} was investigated by
Childs, Schulman and Vazirani~\cite{CSV}. They showed
the following.
\begin{fact}\label{fact:childs}
If $n\geq 2$ and $d$ are constants, then for an $1-o(1)$ fraction of
the hidden polynomials, \HPP$(\F_q,n,d)$ has polylogarithmic
query complexity. 
\end{fact}
Here, like in the case of the \HQPP{}, polynomials are determined
up to constant terms and scalar factors. We are not aware of any
results regarding the quantum computational complexity even 
in the univariate quadratic case. 
For the \HPFGP{}, Decker, Draisma and Wocjan~\cite{DDW09} showed
the following.
\begin{fact}\label{fact:decker}
{\bf (a)} \HPFGP$(\F_q, n,d)$ can be reduced in polynomial time to
\HPFGP$(\F_q, 1,d)$ for every constant $n$. {\bf (b)} 
For every $d$ there exists a finite set $E_d$ of primes such that if 
$d$ is constant  and the
characteristic of $\F_q$ is not in $E_d$ then \HPFGP$(\F_q, 1,d)$ can
be solved in quantum polynomial time.
\end{fact}

\section{A general reduction of the HSSP to the HSP}\label{section:red}
How much greater is the complexity of an $\HSOP$ compared to the
complexity of the related $\HSP$?  To analyze this, we first give a
simple example, which shows that the query complexity of the $\HSOP$
can be exponentially higher than the query complexity of the related
$\HSP$.  Then, more interestingly, we will establish a general
condition on the group action under which the $\HSOP$ can be reduced
in polynomial time to the related $\HSP$.

\subsection{HSSP with exponential query complexity}\label{section:exponential}
While the quantum query complexity of the $\HSP$ is polylogarithmic in
the size of the group, we show in this section that the query
complexity of an $\HSOP$ can be in the order of $|G|^{1/4}$.  More
precisely, we show that Grover's search problem can be reduced to some
specific $\HSOP$.

For a prime power $q$, the {\em general affine group} $\Affq$ of
invertible affine transformations over $\F_q$ is defined as the
semidirect product $\F_q \rtimes \F_q^*$, where $\F_q^*$ denotes the
multiplicative group of $\F_q$.  The natural action of $\Affq$ on
$\F_q$ is defined as $(b,a) \circ x = ax +b$. For every $c \in \F_q$,
the stabilizer of $c$ is the subgroup $H_c= \{ ((1-a)c,a) : a \in
\F_q^*\}$, which has two orbits: $\{c\}$ and $\{d \in \F_q : d \neq c
\}$.  Clearly, $H_c$ is a closed subgroup.  We set ${\cal H} = \{H_c :
c \in \F_q\}.$

\begin{proposition}\label{proposition:grover}
The query complexity of $\HSOP(\Affq, \F_q, \circ , {\cal H})$ is
$\Omega(q^{1/2})$.
\end{proposition}

\begin{proof}
Grover's search over $\F_q$ can be trivially reduced to this
$\HSOP$. Indeed, if the oracle input is $f_c$, defined by $f_c(x) =
\delta_{c,x}$, where $\delta_{c,x}$ is the Kronecker delta, then $f_c$
hides $H_c$ as symmetry subgroup. From any generator $(b,a)$ of $H_c$
one recovers $c$ simply by computing $(1-a)^{-1}b$. Hence, the query
complexity of the $\HSOP$ is at least the query complexity
$\Omega(q^{1/2})$ of Grover's search~\cite{BBBV97}.  
\end{proof}

\subsection{A reduction scheme of the HSSP to the HSP}
In this section, we describe a rather natural framework for reducing
the $\HSOP$ to the related HSP. Essentially, the same idea was used in
\cite{MRRS07} for reducing certain hidden shift problems to the HSP in
the affine group over prime fields. We assume that we are given a
black-box function $f$ over $M$, which hides some subgroup $H$ of $G$
by symmetries. With the help of $f$, we would like to construct a
suitable function $f_{\rm HSP}$ over $G$, which hides $H$.  A first
approach could be to define $f_{\rm HSP}(g)=f(g\circ m)$, where $m$ is
a fixed element of $M$.  Unfortunately, this works only in very
exceptional cases because $f_{\rm HSP}$ takes constant values on the
left cosets of the stabilizer $H_m$ of $m$. Therefore, even in the
simple case when $f$ hides the trivial subgroup, the function $f_{\rm
  HSP}$ will not work unless the stabilizer of $m$ is trivial.  As a
straightforward refinement of this idea, we can pick several elements
$m_1,\ldots,m_t \in M$, and define
\begin{equation*}\label{tuples}
f_{\rm HSP}(g)=( f(g\circ m_1),\ldots, f(g\circ m_t)).
\end{equation*}
For the trivial hidden subgroup, this idea works when the common
stabilizer of $m_1,\ldots,m_t$ is trivial, that is, when
$\bigcap_{i=1}^t H_{m_i}=\{e\}$.  In the theory of permutation groups
such a system of elements is called a {\em base}~\cite{Seress}.  Of
course, bases exist only if the action of $G$ is faithful. The
following definition includes further conditions on $m_1,\ldots,m_t$
in order to make the above construction work in general.

\begin{definition}
\label{strongbase-def}
{\rm Let $G$ be a finite group and let $\circ: G \times M \to M$ be
  an action of $G$ on the finite set $M$. Let $H \leq G$ be a subgroup
  of $G$, and let ${\cal H}$ be a family of subgroups of $G$.  A set
  $B \subseteq M$ is an $H$-{\em strong base} if for every $g \in G$,
  we have
\begin{equation*}
\bigcap_{m \in B} H G_{g \circ m} = H.
\end{equation*}
We call $B$ an ${\cal H}$-{\rm strong base} when it is $H$-strong for
every subgroup $H \in {\cal H}$.  }
\end{definition}
Observe that $\bigcap_{m\in M}HG_m=H^{**}$. Hence, $M$ itself is
always a ${\cal C}(G)$-strong base.  If $B$ is an $H$-strong base,
then $B$ is also an $(x^{-1}Hx)$-strong base for every $x \in G$.
Therefore, if ${\cal H}$ consists of conjugated subgroups, then $B$ is
an ${\cal H}$-strong base if it is an $H$-strong base for some $H \in
{\cal H}$. Also, if $\cal H$ is closed under conjugation by elements
of $G$, $B$ is an ${\cal H}$-{strong base} if and only if $\bigcap_{m \in B} H
G_{m} = H$ for every $H\in {\cal H}$.

The following lemma states that the
$\HSOP$ is indeed reducible to the HSP via an ${\cal H}$-strong base.
\begin{lemma}[Reduction of HSSP to HSP]\label{lemma:reduction}
Let $G$ be a finite group, and let $\circ$ be  an action of $G$ on $M$. Suppose that
the function $f : G \rightarrow  S$ hides some $H \in {\cal H}$ by symmetries. Let $B = \{m_1, \ldots, m_t\}  $ be an
${\cal H}$-strong base. Then $H$ is hidden by the function 
$
f_{\rm HSP}(g)= (f (g \circ m_1),\ldots, f(g \circ m_t))
$.
\end{lemma}

\begin{proof}
We will show that for every $x, y \in G$, we have 
$f_{\rm HSP}(x) = f_{\rm HSP}(y)$ if and only if $y \in Hx$.
To see the "only if" part, suppose that $f_{\rm
  HSP}(x) = f_{\rm HSP}(y)$.  Then by definition $f(x \circ m) = f(y
\circ m)$, for every $m \in B$. Therefore, for every $m\in B$
there exists an element $h_m \in H$ such that $x \circ m = h_m \circ (y \circ m)$. 
This equality implies that $m = (x^{-1} h_m y) \circ m$, that is $x^{-1} h_m y \in
G_m$.  Thus $y \in h_m^{-1} x G_m$, for every $m \in B$, from which we
can deduce $y \in \bigcap_{m \in B} H x G_{m}.$ Now observe that $xG_m
x^{-1} = G_{x \circ m}$, and therefore $y \in \bigcap_{m \in B} H G_{x
  \circ m}x.$ {}From this we can conclude $y \in Hx$ because $B$ is an
${\cal H}$-strong base. To show the reverse implication, suppose that 
$y=h x$ for some $h \in H$. This implies
$y \circ m  = h \circ( x \circ m)$,
for all $m \in B$.  Since $f$ hides $H$ as symmetry subgroup, we have
$f(y \circ m)  = f( x \circ m)$,
again for all $m \in B$, 
implying $f_{\rm HSP}(y)=f_{\rm HSP}(x)$ by the definition of $f_{\rm HSP}$.
\end{proof}

The following statement is immediate from Lemma~\ref{lemma:reduction}.
\begin{proposition}\label{prop:reduction}
Let $G$ be a finite group, $M$ a finite set, $\circ$ a polynomial time computable 
action of $G$ on $M$, and ${\cal H}$ a family of subgroups of $G$.
If there exists an efficiently computable ${\cal H}$-strong base 
in $M$, then $\HSOP(G,M, \circ, {\cal H})$ is polynomial time reducible to $\HSP(G, {\cal H})$.
\end{proposition}

\section{The HSSP for Frobenius complements and the HQPP}
In view of Proposition~\ref{prop:reduction}, we are interested in group
actions for which there exist easily computable (and therefore also
small) bases for some interesting families of subgroups. If in addition
the related $\HSP$ is easy to solve then we have efficiently solvable
$\HSOP$s. It turns out that Frobenius groups under
some conditions not only have these properties, but also that the
$\HQPP$ can be cast as one of these HSSPs.

\subsection{Strong bases in Frobenius groups}
A {\em Frobenius group} is a transitive permutation group 
acting on a finite set such that only the identity element has more than
one fixed point and some non-trivial element fixes a point (see for 
example~\cite{Huppert}).
Let us recall here some notions and facts about these groups.
Let $G$ be a Frobenius group with action $\circ_M$ on $M$. 
The identity element together with the elements of $G$ that have no fixed
points form a normal subgroup $K$, the {\em Frobenius kernel}, for which we also have $|K|=|M|$.
A subgroup $H$ of $G$ is a  {\em Frobenius complement} if it is the stabilizer $H_m$
of some
element $m\in M$. It is a subgroup complementary to
$K$, that is $K\cap H=\{1\}$ and $G=KH$. Hence, the
group $G$ is a semidirect product $K \rtimes H$ 
of $K$ and $H$. 
We define the binary operation $\circ_K : G \times K \rightarrow K$ by
\centerline{
$g\circ_K x=yhxh^{-1},$
}
when $x \in K$ and $g=yh$ with 
$y \in K$ and $h \in H$.
It is a straightforward computation to check that $\circ_K$ is an action of $G$ on $K$.
Furthermore, we
can identify the action $\circ_M$  with the action $\circ_K$ via
the map $\phi: M\to K$ defined as follows.
For any $n \in M$, there exists $g_n \in G$ such that 
$g_n \circ_M m=n$ since $G$ is transitive. If $g_n=y_nh_n$ with 
$y_n \in K$ and $h_n \in H$, by definition we set $\phi (n) = y_n$.
Then indeed
for every $g \in G $ and $n \in M$, we have $g \circ_K \phi(n) = \phi(g \circ_M n)$.
From now on we will suppose without loss of generality that the action is $\circ_K$ which we denote for simplicity by $\circ$.

Observe then that with respect to $\circ$, the Frobenius complement $H$ is the
stabilizer of ${e}$, the identity element of $K$.  
The orbits of $H$ are $\{e\}$ and some other subsets
of $K$, each consisting of $|H|$ elements.
The other  Frobenius complements are $H_x=xHx^{-1}$, for  $x \in K$. 
They are
closed subgroups and their
orbits form closed partitions. 
We denote by ${\cal FC}$ the set of Frobenius complements in $G$.

Let $H$ be the Frobenius complement $H_e$.
Since the Frobenius complements are all conjugates of $H$, 
being an ${\cal FC}$-strong base is equivalent to being
an $H$-strong base. 

To characterize $H$-strong bases it will be convenient 
to use the following notion. For $u,v \in K$ with $u\neq v$, we say 
that $z\in K$ {\em separates} $u$ and
$v$ if $v\circ z\not\in H\circ(u\circ z)$. We have the following
characterization.

\begin{lemma}
\label{Frobenius-base-lemma}
Let $B \subseteq K$. Then $B$ is an $H$-strong base
if and only if for all $u\neq v$ in $K$ 
there exists $z \in B$ which separates $u$ and $v$. 
\end{lemma}

\begin{proof}
To see the "if" part of the statement, suppose that $g' \in \bigcap_{z
  \in B} H G_{g \circ z}$ for some $g',g \in G$. We will prove that
$g' \in H$.  Let $g=yh$ and $g'=y'h'$, where $y,y'\in K$ and $h,h'\in
H$.  Then for every $z \in B$, there exists $h_z \in H$ such that
$g'\circ ( g \circ z) = h_z \circ (g \circ z)$.  Using the definition
of $\circ$, this equality can be rewritten as $y'h' yhz
h^{-1}h'^{-1}=h_z y hzh^{-1}h_z^{-1}$, which is equivalent to
$h^{-1}h'^{-1}y'h' yhz = h^{-1}h'^{-1}h_z hh^{-1}y
hzh^{-1}h_z^{-1}h'h$.  Using again the definition of $\circ$, this is
$h^{-1}h'^{-1}y'h' yh \circ z = h^{-1}h'^{-1}h_zh \circ (h^{-1}yh
\circ z)$.  Set $u=h^{-1}yh$ and $v = h^{-1}h'^{-1}y'h' yh$.  Then for
every $z \in B$, we have $v \circ z \in H \circ (u \circ z)$, that is,
no element in $B$ separates $u$ and $v$.  Therefore, by the assumption
we get $u = v$, which is equivalent to $y' = e$.  Thus $g' =h'$ is
indeed an element of $H$.

To see the reverse implication, assume that there exist $u,v \in K,u\neq v,$
such that none of the elements $z\in B$ separate $u$ and $v$. This
means that for every $z\in B$ there exists an element $h_z\in H$ such
that $v\circ z=h_z\circ (u\circ z)$. Using $v\circ z=(vu^{-1})\circ
(u\circ z)$, this equality implies $vu^{-1}(u\circ z)=h_z\circ (u\circ
z)$, whence $h_z^{-1}vu^{-1}(u\circ z)=u\circ z$, that is,
$h_z^{-1}vu^{-1}\in G_{u\circ z}$.  This gives $vu^{-1}\in
h_zG_{u\circ z}\subseteq H G_{u\circ z}$ for every $z\in B$, that
is, $vu^{-1}\in \bigcap_{z\in B}H G_{u\circ z}$.  As
$vu^{-1}\not\in H$, this contradicts  the definition of a strong
base.
\end{proof}

Our next lemma 
gives a lower bound on the number of elements in $K$ that separate 
$u$ and $v$.

\begin{lemma}
\label{lemma:lower-bound}
Let $|H| \not=|K|-1$. Then for any two distinct elements $u$ and $v$ of $K$
we have \\ \centerline{$ |\{ z \in K : z {\rm ~separates~} u {\rm
    ~and~} v \}| > |K|/2.$} 
\end{lemma}

\begin{proof}

If $z$ does not separate $u$ and $v$ then there exists an element
$h\in H$ such that $vz=huzh^{-1}$ which can also be written as
$hu^{-1}h^{-1}v = hzh^{-1}z^{-1}$.  We say that such an element $h$
{\em belongs} to $z$. The identity element $h=e$ does not belong to
any element $z\in K$ since $u \neq v$. We claim that $h\neq e$ cannot
belong to two distinct elements of $K$. Indeed, if
$hzh^{-1}z^{-1}=hz'h^{-1}z'^{-1}$ then $hz'^{-1}zh^{-1}=z'^{-1}z$,
which in turn implies that $z'^{-1}z=e$ as $e$ is the only element of
$K$ stabilized by the elements of $H$. Therefore, there are at most
$|H|-1$ elements in $K$ which do not separate $u$ and $v$. In other
words, at least $|K|-|H|+1$ of the elements of $K$ separate $u$ and
$v$. Note that $H$ has $(|K|-1)/|H|$ orbits of length $|H|$ on the
nontrivial elements of $K$, and thus $|H|$ divides but is not equal to
$|K|-1$, which implies $|H| \leq (|K| -1)/2$. From this we can indeed
conclude, since then $|K|-|H|+1 > |K|/2$.
\end{proof}

We have the following result regarding  the
existence of {small} strong bases for ${\cal FC}$. 

\begin{proposition}
\label{proposition:frobenius-base}
Let $G$ be a Frobenius group with kernel $K$ such that the cardinality
of the Frobenius complements is different from $|K| - 1$.  Let $B
\subseteq K$ be a uniformly random set of size $\ell$, where $\ell =
\Theta(\log |K| \log 1/ \epsilon)$.  Then $B$ is an ${\cal FC}$-strong
base with probability of at least $1-\epsilon$.
\end{proposition}

\begin{proof}
Let $B$ be a uniformly random subset of $K$ of size $\ell$.  By
Lemma~\ref{Frobenius-base-lemma} it is sufficient to prove that with a
probability of at least $1-\epsilon$, for every $u \neq v$, there exists
an element in $B$ which separates $u$ and $v$.  We will in fact upper
bound the probability of the opposite event.  For a fixed pair $u \neq
v$, by Lemma~\ref{lemma:lower-bound}, the probability that a random $z$
does not separate $u$ and $v$ is at most 1/2. Therefore, the
probability that none of the elements in $B$ separates $u$ and $v$ is
less than $2^{-\ell}$. Thus, the probability that for some pair $u \neq
v$ none of the elements in $B$ separates $u$ and $v$ is less than
$\binom{|K|}{2}2^{-\ell}$, which is at most $\epsilon$ by the choice
of $\ell$.  
\end{proof}

If $G$ is a Frobenius group that satisfies the condition of
Proposition~\ref{proposition:frobenius-base} then we can compute
efficiently a small base for the Frobenius complements, because there
are efficient algorithms for random sampling nearly uniformly in
black-box groups~\cite{babai}.  Therefore, by
Proposition~\ref{prop:reduction} we can efficiently reduce the $\HSOP$
to the related HSP and we obtain the following result.

\begin{theorem}\label{theorem:frobenius-reduction}
Let $G = K \rtimes H $ be a Frobenius group with action $\circ$ 
such that $ |H| < |K|-1$.  Then $\HSOP(G, K, \circ, {\cal FC})$
is reducible in probabilistic polynomial time to $\HSP(G, {\cal FC})$. 
\end{theorem}

We remark that the reduction of Grover's search to a specific HSSP in
Proposition~\ref{proposition:grover} can be extended to arbitrary
Frobenius groups when $|H|=|K|-1$, that is, sharply 2-transitive
groups.  Therefore, for such groups it not only follows that small
$H$-bases fail to exist but it also follows that even the quantum
query complexity of the \HSOP{} is $\Omega(|G|^{1/4})$.  Also, the
only strong base in a sharply 2-transitive group is the whole $K$.

\subsection{Affine groups}
\label{section:affine}
As any affine group, the general affine group $\Affq = \F_q \rtimes
\F_q^*$ defined in Section~\ref{section:exponential} is a Frobenius
group.  Its kernel is $\F_q$. In the terminology of Frobenius groups,
we have proved in Proposition~\ref{proposition:grover} that for $\Affq$ the 
$\HSOP$ for the complements is difficult. Let
$H$ be a proper subgroup of $\F_q^*$ which is not the trivial
group. We define the group $\Affq(H)$ as $\F_q \rtimes H$. With the
restriction of the natural action, denoted here by $\circ$, $\Affq(H)$
is also a Frobenius group.  In contrast to the difficulty in the full
affine group, we obtain the following positive results for the smaller Frobenius
groups. They are consequences of the analogous results
for the related $\HSP$ stated in Facts~\ref{fact:ettinger}
and~\ref{fact:friedl}, via the reduction of
Theorem~\ref{theorem:frobenius-reduction}. 
Statements
(a) and (b) are not new, they are proved in a slightly different formulation
in \cite{MRRS07}, using implicitly the randomized construction
for a strong base. 
For (c) note that the
derived subgroup in $\Affq(H)$ is indeed commutative.

\begin{corollary}\label{corollary:affine}
Let $q$ be a prime power and let $H \leq \F_q^*$ such that $1 < |H| < q-1$. 
The following results hold for $\HSOP(\Affq(H), \F_q, \circ, {\cal FC})$:
\begin{itemize}
\item[(a)]
It has polynomial query complexity.
\item[(b)]
It can be solved in quantum polynomial time when $q$ is prime and 
$|H| = \Omega (q/ {\rm polylog}(q))$.
\item[(c)]
It can be solved in quantum polynomial time when $q$ is the power of a fixed prime.
\end{itemize}
\end{corollary}

The case of $\Affq(\{\pm 1\})$, when $q$ is an odd prime power is
particularly interesting.  It turns out that the $\HSOP$ over
$\Affq(\{\pm 1\})$ for the Frobenius complements is essentially the
same problem as the $\HQPP$ over $\F_q$.

\begin{proposition}\label{proposition:hqp}
The following problems are polynomially equivalent:
\begin{enumerate}
\item 
$\HQPP(\F_q)$
\item $\HSOP(\Affq(\{\pm 1\}), \F_q, \circ, {\cal FC})$
\item 
$\HSP(\Affq(\{\pm 1\}),{\cal FC})$
\end{enumerate}
\end{proposition}

\begin{proof}
The first two problems are equivalent as we claim that every
$f : \F_q \rightarrow S$, as oracle input for
$\HQPP(\F_q)$ hides the polynomial ${\cal P}_u$ if and only if as oracle input for
$\HSOP(\Affq(\{\pm 1\}), \F_q, \circ, {\cal FC})$  it hides the
Frobenius complement $H_u$. To see this, observe that
the level sets of ${\cal P}_u$ are of the form $\{x+u, -x +u\}$,
which are exactly the orbits of $H_u$. Therefore, we have the 
following equivalences:
\begin{eqnarray*}
f {\rm ~hides~} p_u 
& 
\Longleftrightarrow 
& 
\pi_f = \{ \{x+u, -x +u\} : x \in \F_q \} 
\\
& 
\Longleftrightarrow 
& 
\pi_f^* = H_u 
\\
& 
\Longleftrightarrow 
& 
f {\rm ~hides~} H_u {\rm ~by~symmetries}
\end{eqnarray*}
The reduction from the second problem to the third one is provided by
Theorem~\ref{theorem:frobenius-reduction}.  Note that we can construct
a base deterministically by choosing two different elements of order
two. For a reduction in the reverse direction, consider a function $f$
on $\Affq(\{\pm 1\})$ which hides the subgroup $H_{u}=\{(0,1),(2u,-1)
\}$. Then all the collisions taken by $f$ on elements of $\Affq(\{\pm
1\})$ are $f(2u-b,-1)=f(b,1)$ for $b\in \F_q$. We define a new
function $f^\circ$ on $\F_q$ as
$f^\circ(b)=\min\left(f(b,1),f(b,-1)\right)$.  Examining the possible
collisions gives that for $b\neq b'\in \F_q$ we have
$f^\circ(b)=f^\circ(b')$ if and only if $b'=2u-b=(2u,-1)\circ b$.
\end{proof}

Together with
Corollary~\ref{corollary:affine} (c) the statements of this proposition
imply the following result.

\begin{theorem}\label{theorem:hqp-small}
$\HQPP(\F_q)$ is solvable 
in quantum polynomial time over 
  constant characteristic fields.
\end{theorem}
We observe that in contrast to the constant characteristic case, the
\HQPP{} appears to be difficult over prime fields $\F_p$, as it is
equivalent to the \HSP{} in the dihedral group $D_{2p}\cong
\Affp(\{\pm 1\})$.

Note that in~\cite{vDHI03} van~Dam, Hallgren and Ip gave a polynomial
time solution to a problem which can be considered as a version of
$\HSOP(\Affq(H), \F_q, \circ, {\cal FC})$ where the function hiding
the complement is promised to be a shifted multiplicative character
$\chi:\F_q^*\rightarrow \C^*$. This strong promise (in our oracle
model we can only check for equality of the output values)  makes the problem
efficiently solvable even in the case $H=\{\pm 1\}$ where the $\HSOP$
with general hiding function appears to be difficult.

We also remark that strong bases in the Frobenius group $\Affq(H)$
with $|H|=(q-1)/2$ play an important role (under the name {\em
  factoring sets}) in certain algorithms for factoring univariate
polynomials over $\F_q$, see~\cite{Camion83}. This is because a set
$B$ which separates two (unknown) field elements $u$ and $v$ can be
used to find a proper decomposition of a polynomial having both $u$
and $v$ as roots.  In fact, an efficient deterministic construction of
strong bases for such affine groups over prime fields would imply an
efficient deterministic algorithm for factoring polynomials over finite
fields.

\subsection{Multivariate quadratic hidden polynomials}
In this part, we reduce the $\HPP$ for multivariate polynomials of
degree at most two to the univariate $\HQPP$.  As already noted,
adding a constant term does not change the level sets, therefore we
consider polynomials with zero constant term. Thus, we assume that the
hidden polynomial is of the form
\begin{equation}\label{multipolynom}
{\cal P}(x_1,\ldots,x_n)= \sum_{1\leq i\leq j\leq n}a_{ij}x_{i}x_j+
\sum_{1\leq k \leq n}b_kx_k.
\end{equation}
Also, as the partition $\pi_{\cal P}$ remains the same
when we multiply all coefficients with the same non-zero
element from $\F_q$, we consider that the $\HPP$ has been solved
if we determine the ratios between all the pairs
of the $n(n+1)/2$ coefficients $a_{ij}$ and $b_k$.

\begin{proposition}
The problem $\HPP(\F_q,n,2)$ can be reduced on a quantum computer
to $O(n^2)$ instances of $\HQPP(\F_q)$ in time $(n+\log q)^{O(1)}$.
\end{proposition}

\begin{proof}
In order to
simplify the following discussions we define $a_{ji}$ to be $a_{ij}$
for $j>i$. Additionally, if $q=2$ then we also assume
$a_{ii}=0$ because $x^2=x$ holds over $\F_2$. We assume 
that we have a procedure
$\cal R$ for determing the coefficients of a univariate quadratic
polynomial up to a common factor. Its oracle input is a function
on $\F_q$ that has the same level set structure as a polynomial of
the form $ax^2+bx$.  We assume that $\cal R$ decides whether $a$ is
zero and if $a\neq 0$ then $\cal R$ returns the quotient $b/a$.

We start with the case $n=2$. We have an oracle with the same
level sets as the polynomial
\[
{\cal P}(x_1,x_2)=a_{11}x_1^2+a_{22}x_2^2+a_{12}x_1x_2+b_1x_1+b_2x_2\,.
\]
We use the oracle with the inputs $(x_1,x_2):=(x,0)$. This
way, we obtain an instance of  $\HQPP$ for the
univariate polynomial $a_{11}x^2+b_1x$. We use $\cal R$ to decide
whether $a_{11}$ is zero or not and if $a_{11}\neq 0$ then we compute the
quotient $b_1/a_{11}$. Furthermore, we set $(x_1,x_2):=(x,1)$ for the
inputs of the oracle to compute $(a_{12}+b_1)/a_{11}$ in the second
step. From this result we can easily compute the quotient
$a_{12}/a_{11}$.  Similarly, using the substitutions
$(x_1,x_2):=(0,x)$ and $(x_1,x_2):=(1,x)$ we decide whether $a_{22}$
is zero or not.  If $a_{22}\neq 0$ then we obtain the quotients
$a_{12}/a_{22}$ and $b_2/a_{22}$. We now consider the following
different cases.

\begin{itemize}
\item $a_{11},a_{22}\neq 0$: If $a_{12}\neq 0$ then we have determined
   all coefficients of $\cal P$ up to a common factor. If $a_{12}=0$ then we
  use the inputs $(x_1,x_2):=(x,x)$ and we obtain $\HQPP$ for
   $(a_{11}+a_{22})x^2+(b_1+b_2)x$. With $\cal R$ we can determine
   whether $a_{11}+a_{22}$ is zero or not. If it is non-zero then we
   find an element $r\in \F_q$ such that
    $b_1+b_2=r(a_{11}+a_{22})$. When we write $b_i/a_{ii}=c_i$ then the
    equation $(r-c_1)a_{11}=(c_2-r)a_{22}$ follows. Since $a_{ii}\not =
    0$, we can compute easily all coefficients of $\cal P$ up to a common
    factor. If $a_{11}+a_{22}=0$ then we also can compute all coefficients
    easily. 

\item $a_{11}\neq 0,a_{22}=0$: If $a_{12}=0$ then we use the inputs
  $(x_1,x_2):=(x,x)$ and we obtain $\HQPP$ for the polynomial
  $a_{11}x^2+(b_1+b_2)x$. With $\cal R$ we can determine the quotient
  $(b_1+b_2)/a_{11}$ and together with the already known value
  $b_1/a_{11}$ we obtain the missing $b_2/a_{11}$.  If $a_{12}\neq 0$
  then we pick $\alpha\in \F_q\setminus\{0\}$ such that $1+\alpha
  a_{12}/a_{11}\neq 0$ and we use the inputs $(x_1,x_2):=(x, \alpha
  x)$. We obtain $\HQPP$ for $(a_{11}+\alpha a_{12})x^2+(b_1+\alpha
  b_2)x$, which can be used to find $r\in \F_q$ such that $(b_1+\alpha
  b_2)=r(a_{11}+\alpha a_{12})$. This gives us the missing fraction
  $b_2/a_{11}$.  The case $a_{22}\neq 0$ and
  $a_{11}=0$ can be treated in a similar way.

\item $a_{11}=a_{22}=0,q\not=2$: We use the inputs $(x_1,x_2):=(x,x)$
  and obtain $\HQPP$ for the polynomial $a_{12}x^2+(b_1+b_2)x$ that
  can be used to decide whether $a_{12}=0$ or not. If it is non-zero
  then we compute $(b_1+b_2)/a_{12}$. Furthermore, we can choose
  $\alpha\in\F_q^\times, \alpha \not =1$, and we use the inputs
  $(x_1,x_2):=(x,\alpha x)$ to compute the fraction $(b_1+\alpha
  b_2)/(\alpha a_{12})$.  From these two fractions we can determine
  $b_1/a_{12}$ and $b_2/a_{12}$. If $a_{12}=0$ then we have the
  polynomial $b_1x_1+b_2x_2$ and we can determine the ratio between
  $b_1$ and $b_2$ by the algorithm for the abelian HSP over the
  additive group of $\F_q^2$. Note that we use a quantum computer
  for an efficient implementation of this step of the reduction.

\item $a_{11},a_{22}=0,q=2$: We use the inputs $(x_1,x_2):=(x,0)$ and
  we obtain $\HPP$ for the polynomial $b_1x$. We can easily test
  whether it is constant, i.e. $b_1=0$, or not. The coefficient $b_2$
  can be computed in a similar way. The input $(x_1,x_2):=(x,1)$ give
  us $a_{12}+b_1$.

\end{itemize}

This shows that we can find all coefficients of a bivariate polynomial
up to a common factor when we use $\cal R$ only a constant number of
times and some additional operations, which can be performed
efficiently on a quantum computer. 

Next we consider the case $n=3$. Substituting zero in $x_3$, we can
use the algorithm for the bivariate case to test whether
$a_{11}=0$. If $a_{11}\neq 0$ we can determine the quotient of the
remaining coefficients (except for $a_{23}$) and $a_{11}$ by
substituting zero in $x_2$ or $x_3$ and using the algorithm for the
bivariate case. For $a_{23}$ we can substitute $(x_1,x_2,x_3)=(x,y,y)$ and
obtain the polynomial 
\[
a_{11}x^2 + (a_{12}+a_{13})xy + (a_{22}+a_{23}+a_{33}) y^2 +b_1x+(b_2+b_3)y \,.
\]
Then the algorithm for the bivariate case gives us
$(a_{22}+a_{23}+a_{33})/a_{11}$ from which we can compute
$a_{23}/a_{11}$. The cases where any of the coefficients
$b_1$, $a_{22}$, $b_2$, $a_{33}$, or $b_3$ is non-zero can be treated
in a similar way. It remains to handle the case of a polynomial of the
form $a_{12}x_1x_2+ a_{13}x_1x_3+a_{23}x_2x_3$. Then substituting $1$
in $x_3$ gives the polynomial $a_{12}x_1x_2+ a_{13}x_1+a_{23}x_2$ and
the ratio between the three coefficients can be found by the bivariate
algorithm.

The case $n=4$ can be handled as follows. We apply the algorithm of
the preceding paragraph to the four polynomials obtained by
substituting zero in $x_1$, $x_2$, $x_3$, and $x_4$, respectively.
Observe that these steps determine the ratio between pairs of
coefficients that have indices that fit in a three-element subset of
$\{1,2,3,4\}$.  By transitivity, we are done unless our polynomial is
of the form $a_{12}x_1x_2+a_{34}x_3x_4$, $a_{13}x_1x_3+a_{24}x_2x_4$,
or $a_{14}x_1x_4+a_{23}x_2x_3$. If it
is of the form $a_{12}x_1x_2+a_{34}x_3x_4$ then we can determine the
ratio between the coefficients by using the bivariate algorithm by
substituting $x_1$ in $x_2$ and $x_3$ in $x_4$. The two remaining
polynomials can be treated in a similar way. 

Finally we consider the case $n>4$. Using $O(n^2)$ applications of the
bivariate algorithm, we find indices $i\neq j$ such that at least one
of $a_{ii}$, $b_i$ or $a_{ij}$ is non-zero. The ratio between this
coefficient and any other can be computed using the algorithm for two,
three, or four variables. The cost of these steps amounts to $O(n^2)$
applications of the procedure $\cal R$ and a polynomial number of
other operations.
\end{proof}

\begin{theorem}\label{theorem:multivariate}
$\HPP(\F_q,n,2)$ can be solved by a polynomial time quantum algorithm
  over fields of constant characteristic.
\end{theorem}

\section{Function graph groups and the HPGP}\label{hpfgp}
For dealing with the HPGP we define a family of semidirect product
groups that we call function graph groups. We show that each instance
of the ${\HPFGP}(\F_q,1,d)$ can be reduced to the $\HSP$ for an
appropriate function graph group corresponding to univariate
polynomials of degree at most $d$.  These special function graph groups
are semidirect products of groups of $q$-power order. Therefore, they
cannot be Frobenius groups.

\subsection{The HPGP as HSSP over function graph groups}\label{subsec hpfgp}
It will be convenient to work in a more general setting. 
\begin{definition}
{\rm Let $A$ and $B$ be two abelian groups.  The family of functions
  mapping $A$ to $B$ forms an abelian group ${\cal F}$ with the addition
  defined as $(Q_1+Q_2)(x)=Q_1(x)+Q_2(x)$.  For every $t\in A$, the
  shift map $a_t$ defined as $(a_tQ)(x)=Q(x-t)$ is an automorphism of
  this group.  A {\em function group} from $A$ to $B$ is a subgroup
  $K$ of ${\cal F}$ which is closed under the shift maps.  We denote
  the restriction of $a_t$ to $K$ also with $a_t$. Then the map
  $t\mapsto a_t$ is a homomorphism from $A$ to the automorphism group
  of $K$.  The {\em function graph group} $\Fg(K)$ is defined as the
  semidirect product ${K}\rtimes_{t\mapsto a_t} A$.  }
\end{definition}
The multiplication of $\Fg(K)$ is given by the rule
\begin{equation*}\label{def concat}
(Q_1,t_1)(Q_2,t_2)=(Q_1+a_{t_1}Q_2,t_1+t_2).
\end{equation*}
The  {\em shifting action} $\circ$ of $\Fg(K)$ on $A\times B$ is defined as
\begin{equation*}\label{eq action}
(Q,t)\circ(x,y)=(x+t,y+Q(x+t)).
\end{equation*}
For $t\in A$ and $Q\in K$, we set $a_{Q,t}=(Q-a_tQ,t)$, the conjugate
of the element $(0,t)$ by $(Q,0)$. Furthermore, let $A_{Q}=\{a_{Q,t}:
t\in A\}$ be the conjugate of the subgroup $\{(0,t) : t\in A\}$ by
$(Q,0)$.  Then every $A_Q$ is a subgroup of $\Fg(K)$ that is
complementary to the normal subgroup $\{(Q,0) : Q\in K\}$. We call
them {\em standard complements}, and we denote by ${\cal SC}$ the
family $\{A_Q : Q \in K\}$ of the standard complements.

We are now ready to show a connection between function graph problems
and the orbits of the standard complements in function graph groups.

\begin{proposition}\label{proposition:hpfgp-hsop}
Let $\Fg(K)$ be a function graph group, let $\circ$ be its shifting
action on $A\times B$, and let $A_Q$ be a standard complement. Then
$A_Q$ is closed and
the orbits of $A_Q$ are the level sets of the function $f :
(x,y)\mapsto y-Q(x)$ on $A\times B$.
\end{proposition}

\begin{proof}
Assume that $A_Q$ is not closed. Then, as $A_Q$ is a complement of
$\{(Q',0) : Q'\in K\}$, there exists $Q'\in K\setminus \{0\}$ such
that $(x,y+Q'(x))=(Q',0)\circ(x,y)\in A_Q\circ(x,y)$ for every pair
$(x,y)\in A\times B$. This is a contradiction since
$a_{Q,t}(x,y)=(x,y')$ is only possible if $t=0$ and $y'=y$.  

To see
the second part of the statement, observe that $f(x,y) = f(x',y')$ iff
$\exists t \in A: ~ (x' , y') = (x+t, y-Q(x)+Q(x+t))$ iff $\exists t
\in A: ~ (x' , y') = a_{Q,t}\circ(x,y)$ iff $(x' , y') \in A_Q \circ
(x,y)$.
\end{proof}

We now specialize function graph groups to polynomials which relate
them to the \HPFGP{}.  Let $A$ and $B$ be the additive group of $\F_q$
and let $K$ be $\Fqd$, the set of polynomials of degree at most
$d$. Observe that we include also polynomials with non-zero constant
terms in order to be closed under the shifts. Then
Proposition~\ref{proposition:hpfgp-hsop} translates to the following
statement.
\begin{proposition}\label{proposition:idontknow}
Let $f : \F_q \times \F_q\rightarrow S$ be a function. Then $f$ hides for
$\HPFGP(\F_q, 1, d)$ the polynomial $Q$ if and only if for
$\HSOP(\Fg(\Fqd), \F_q \times \F_q, \circ, {\cal SC})$ it hides the
standard complement $A_Q$ by symmetries.
\end{proposition}

\subsection{Small bases for standard complements}
In this section, we construct strong bases for the standard
complements in function graph groups. The next lemma gives a simple
characterization of such bases.

\begin{lemma}\label{lemma:sc-strong-base}
Let $\Fg(K)= K \rtimes A$ be a function graph group with action
$\circ$ on $A \times B$.  Let $D=\{ (x_1,y_1), \ldots, (x_\ell,y_\ell)
\}$ be a subset of $A \times B$.  Then $D$ is an ${\cal SC}$-strong
base if and only if for all $Q \in K$, the equation $Q(x_1)=\ldots =
Q(x_\ell)=0$ implies $Q=0$.
\end{lemma}

\begin{proof}
As ${\cal SC}$ is closed under conjugation, by the remarks following
Definition~\ref{strongbase-def}, $D$ is an ${\cal SC}$-strong base if
and only if $\bigcap_{i=1}^\ell A_Q\Fg(K)_{(x_i,y_i)}=A_Q$ for every
$Q\in K$. The statement $(Q',t')\in A_Q\Fg(K)_{(x_i,y_i)}$ is true if
and only if there is a $t_i\in A$ such that $
(Q',t')\circ(x_i,y_i)=a_{Q,t_i}\circ(x_i,y_i)\,.  $ This can be
rewritten as $ (x_i+t',y_i+Q'(x_i+t') )=
(x_i+t_i,y_i-Q(x_i)+Q(x_i+t_i))\,.  $ The equality holds if and only
if $t_i=t'$ and $(a_{-t'}Q'-a_{-t'}Q+Q)(x_i)=0$. Hence, an element
$(Q',t')$ is in the intersection $\bigcap_{i=1}^\ell
A_Q\Fg(K)_{(x_i,y_i)}$ iff $t_i=t'$ and
$(a_{-t'}Q'-a_{-t'}Q+Q)(x_i)=0$ holds for all $i$.

We first prove the "only if" part
of the lemma. To this end, let $D$ be an $A_Q$-strong base and let $R \in
{\cal C}$ be a function such that $R(x_i)=0$ for all $i$. Let
\[
(Q',t')\in \bigcap_{i=1}^\ell A_Q\Fg(K)_{(x_i,y_i)}
\]
be any element. Since $D$ is a base we know that the intersection is
equal to $A_Q$ and from this $(Q',t')=(Q-a_{t'}Q,t')$ follows. We also
know that $(a_{-t'}Q'-a_{-t'}Q+Q)(x_i)=0$ for all $i$ and the same is
true for $a_{-t'}Q'+R-a_{-t'}Q+Q$. Hence, we also have $(Q'+a_{t'}
R,t')$ in the intersection and $(Q'+a_{t'}R,t')=(Q-a_{t'}Q,t')$
follows. We have $(Q',t')=(Q'+a_{t'}R,t')$ and this directly implies
$a_{t'}R=R=0$.

To see the "if" part of the lemma, observe that the second statement
of the lemma, applied to the function $a_{-t'}Q'-a_{-t'}Q+Q$, implies
that $a_{-t'}Q'-a_{-t'}Q+Q=0$.  We apply the shift map $a_{t'}$ to
this equality and we obtain $Q'=Q-a_{t'}Q$. Hence, we have
$(Q',t')=a_{Q,t'}\in A_Q$ and this shows that $D$ is an ${\cal
  SC}$-strong base.
\end{proof}

Combining the statements of this section we obtain the following
result.

\begin{theorem}\label{theorem:hpgptohsp} 
${\HPFGP}(\F_q,1,d)$ can be reduced to ${\HSP}(\Fg(\Fqd), {\cal SC})$
  in polynomial time in $d$ and $ \log q$.
\end{theorem}

\begin{proof}
Univariate polynomials of degree $d$ have at most $d$ roots over a
field. Therefore, by Proposition~\ref{proposition:idontknow},
Lemma~\ref{lemma:sc-strong-base} and Lemma~\ref{lemma:reduction}, we
can associate in polynomial time an instance of ${\HPFGP}(\F_q,1,d)$
that hides a polynomial $Q$ with symmetries to an instance of
${\HSP}(\Fg(\Fqd), {\cal SC})$ that hides the subgroup $A_Q$.  Then
the polynomial $Q$ (up to a constant term) can be recovered from
generators for $A_Q$ as follows. The elements
$(Q-a_{t_1}Q,t_1),\ldots,(Q-a_{t_\ell}Q,t_\ell)$ generate $A_Q$ if and
only if $t_1,\ldots,t_\ell$ generate the additive group of $\F_q$. It
follows that for arbitrary $s\in \F_q$, we can efficiently compute
$Q-a_sQ$ using the group operation in $A_Q\leq \Fg(
\F_q^{(d)}[x])$. Substituting $s$ into $Q-a_sQ$ gives $Q(s)-Q(0)$. We
do this for $d$ different values $s\in \F$ and compute $Q-Q(0)$ using
Lagrange interpolation.
\end{proof}

We remark that the group $\Fg(\Fqd)$ is of nilpotency class
$d+1$. However, we can actually give a reduction to the $\HSP$ in a
group of class $d$. To this end, observe that the hidden subgroup is a
conjugate of the complement $\F_q$.  Therefore, it can be found in the
subgroup generated by the commutator of $\F_q$ with $\Fqd$ (this is an
abelian normal subgroup) and the complement $\F_q$. This semidirect
product group has nilpotency class $d$.

Note that semidirect product groups $\Fg(\Fqd)$ are metabelian and
that their exponent is the characteristic of $\F_q$. Therefore, for
fixed characteristic, we can apply Fact~\ref{fact:friedl}~(b) to
obtain the following result.

\begin{corollary}\label{cor:noExceptions}
Assume that $q$ is a power of a fixed prime $p$.  Then
${\HPFGP}(\F_q,1,d)$ can be solved by a quantum algorithm in time
polynomial in $d$ and $ \log q$.
\end{corollary}

This corollary allows us to complete Fact~\ref{fact:decker} (b),
because it can be applied to fields of characteristic in the set $E_d$
that were left open.  Since $E_d$ is finite it follows that for fixed
$d$ we can solve ${\HPFGP}(\F_q,1,d)$ in quantum polynomial time for
all finite fields. Together with Fact~\ref{fact:decker} (a) this
improves the overall result of~\cite{DDW09}: the ${\HPFGP}(\F_q,n,d)$
can be solved efficiently for all finite fields when $n$ and $d$ are
constant.  We further improve this result in the next section where we
present a more powerful reduction of the multivariate problem to the
univariate case.

We conclude this section by showing that the HSP for semidirect
product groups of the form $\Z_p^m \rtimes \Z_p$ can be reduced to a
multidimensional analogue of the \HPFGP{}. This HSP is discussed
in~\cite{BCvD05} and it is shown there that the HSP for all possible
subgroups can be reduced to the HSP where the hidden subgroups are
complements of $\Z_p^m$. Let $H$ be such a subgroup. Following
arguments of~\cite{BCvD05}, we show that the cosets of $H$ can be
considered as level sets of a polynomial map  from $\Z_p^{m+1}$ to
$\Z_p^m$ of the form $y-Q(x)$, where $y=(y_1,\ldots,y_m)$ and
$Q(x)=(Q_1(x),\ldots,Q_m(x))$, each $Q_i(x)$ being a univariate
polynomial over $\Z_p$ of degree at most $d$. Here $d\leq {\rm
  min}(m,p)$ depends on the structure of $G$ (actually its nilpotency
class).

To this end, notice that the semidirect product structure is given by
a linear transformation $A$ on $\Z_p^m$.  (This is the action of the
generator $1$ of $\Z_p$ on $\Z_p^m$.) We have $A^p=I$, whence $B=A-I$
satisfies $B^p=(A-I)^p=A^p-I^p=0$. Therefore, there exists a smallest
positive integer $d\leq {\rm min}(m,p)$ such that $B^d=0$.

A subgroup $H_v$ complementary to $\Z_p^m$ in $\Z_p^m\rtimes \Z_p$ 
consists of the powers of an element of the form $(v,1)$ 
for some $v\in \Z_p^m$. With the
map
\begin{equation}\label{def qv}
Q_v:\left\{ \begin{array}{l}\Z_p \to \Z_p^m \\ t\mapsto
\sum_{j=0}^{t-1}A^jv\end{array}\right.
\end{equation}
these powers are the pairs $(Q_v(t),t)$, for $t\in \Z_p$, and the right
cosets of $H_v$ are the sets of the pairs $(Q_v(t)+y,t)$, for $t\in
\Z_p$, where $y\in \Z_p^m$. It turns out 
that the entries of the matrix of $\sum_{j=0}^{t-1}A^j$,
as functions in $t$, are polynomials of degree at most $d$ with
zero constant term (see~\cite{BCvD05}). Therefore,
the same holds for the coordinates $Q_v^{(i)}$ of the vector $Q_v(t)$.
In other words, the map $t\mapsto
Q_v(t)$ is a polynomial map from $\Z_p$ to $\Z_p^m$ of degree $d$ with
zero constant term. Hence, the cosets of $H_v$ are
exactly the level sets for the polynomial map
\begin{equation}\label{eq problem}
(y_1,\ldots, y_m,x)\mapsto (y_1-Q_v^{(1)}(x),\ldots,y_m-Q_v^{(m)}(x))
\end{equation}
from $\Z_p^{m+1}$ to $\Z_p^m$.

It follows that any function on $\Z_p^m\rtimes\Z_p$
that hides the subgroup 
$H_v=\langle (v,1)\rangle$ directly defines an
instance of the $m$-dimensional analogue of
the \HPFGP{} for $Q_v$ as defined in Eq.~(\ref{def qv}). 
If we solve the $m$-dimensional \HPFGP{} for these instances, i.e., if 
we determine $Q_v$, then we obtain $v$ by calculating $v=Q_v(1)$. 

This shows that the HSP of $\Z_p^m \rtimes \Z_p$ can be indeed
efficiently reduced to the $m$-dimensional analogue of the \HPFGP{}.
Plugging $A=\Z_p$, $B=\Z_p^m$ and $K=(\Z_p^{(d)}[x])^m$ into 
Proposition~\ref{proposition:hpfgp-hsop}, we obtain that
this problem can be viewed as an instance of the HSSP
over a semidirect product $K$ with $\Z_p$. 
Here the functions are vectors of univariate polynomials. Therefore,
by Lemma~\ref{lemma:sc-strong-base}, small bases exist and can
be found easily and the reduction to a HSP works. Note, however,
that the new group is in general much bigger than the 
original one. 

These reductions explain why it was possible to construct the algorithm 
of~\cite{DDW09} in close analogy with the pretty good measurement
framework of~\cite{BCvD05} for semidirect product groups.

\subsection{Reduction of multivariate HPGP to univariate case}
The scheme of~\cite{DDW09} for reducing the multivariate HPGP to
the univariate case can be improved with the help of a generalized
Vandermonde matrix.
\begin{theorem}\label{theorem:new-reduction}
An instance of ${\HPFGP}(\F_q,n,d)$ can be reduced to $O\binom{d+n}{n}$
 instances of ${\HPFGP}(\F_q,1,d)$ by a classical algorithm with  running time
polynomial in $\binom{d+n}{n}$. If $d$ is constant then
${\HPFGP}(\F_q,n,d)$ can be solved 
by a polynomial time quantum algorithm.
\end{theorem}

We prove the theorem in the remainder of the subsection.
For this, we consider $n$-variate polynomials of the special form
\begin{equation}\label{eq hpfgp}
y-Q(x_1,\ldots,x_n)\quad {\rm with} \quad Q \in
\F_q[x_1, \ldots, x_n]\,.
\end{equation}
Note that we changed the notation by replacing the polynomials
$Q(x_1)$ in Def.~\ref{def 1} by polynomials $Q(x_1,\ldots, x_n) \in
\F_q[x_1,\ldots,x_n]$.  This makes the following discussion easier.  Recall
that the constant term of the polynomials $Q(x_1,\ldots,x_n)$ is
assumed to be zero since it cannot be determined.  Furthermore, the
identity $x^q=x$ in $\F_q$ implies that we can only distinguish
polynomials that are reduced modulo $x_i^q-x_i$ for all variables
$x_i$.  Hence, for a maximum total degree $d$ we only consider local
degrees of at most $\min\{d,q-1\}$, i.e., the power of each $x_i$ in
all monomials occuring in $Q(x_1,\ldots,x_n)$ is less or equal to this
minimum.

For each $j$ with $1\le j\le n$ let
\[
\mathcal{I}^{(j)} := \left\{\alpha\in\N^j \, : \, \sum_{i=1}^j
\alpha_i \le d,\, \alpha_i \le \min\{d,q-1\}\mbox{ for
  $i=1,\ldots,j$}\right\} \setminus \{(0,\ldots,0)\}
\]
be the set of all exponent vectors for the monomials of total degree
at most $d$ when the variables are restricted to $x_1,\ldots, x_j$.
For each $\alpha\in\mathcal{I}^{(j)}$ let
\[
m_\alpha:= x_1^{\alpha_1} \, \cdots x_j^{\alpha_j}
\]
denote the corresponding monomial. For $j$ and $j'$ with $1\le j < j'
\le k$, a monomial $m_\alpha$ with $\alpha=(\alpha_1,\ldots,\alpha_j)
\in \mathcal{I}^{(j)}$ is also defined by
$\tilde{\alpha}=(\alpha_1,\ldots,\alpha_j,0,\ldots,0) \in
\mathcal{I}^{(j')}$.  Finally, for $v=(v_1,\ldots,v_j)\in\F_q^j$ let
\[
m_\alpha(v) := v_1^{\alpha_1} \cdot \ldots \cdot v_j^{\alpha_j}
\]
denote the evaluation of the monomial $m_\alpha$ at the point $v$.
For $q>\ell$, the number of such monomials is given by the simple
expression
\[
| \mathcal{I}^{(j)} |= \binom{d + j}{j} - 1\,.
\]
For $q\le\ell$, the number of such monomials is determined with the
inclusion-exclusion principle, which leads to the expression
\[
|\mathcal{I}^{(j)}| = \sum_{i=0}^j (-1)^i \binom{j}{i}
\binom{d - i  q + j}{j} - 1\,.
\]
We use the convention that the binomial coefficient is zero if the
number at the top is negative. With the help of $\mathcal{I}^{(j)}$
we can define the generalized Vandermonde matrix and describe an
efficient construction.

\begin{lemma}\label{lemma:reduce}
Let $d$ be the maximum total degree of the monomials in
$\mathcal{I}^{(j)}$ over the field $\F_q$. Then there
is a  classical algorithm for constructing a set
$\mathcal{V}^{(j)} \subset \F_q^{j}$ of cardinality
$|\mathcal{I}^{(j)}|$ such that the square matrix
\begin{equation}
M^{(j)} := \Big[ \, m_\alpha(v) \,
  \Big]_{v\in\mathcal{V}^{(j)},\,\alpha\in
  \mathcal{I}^{(j)}}
\end{equation}
has full rank.  This matrix is called the generalized Vandermonde
matrix.  The running time of the algorithm is polynomial in
$|\mathcal{I}^{(j)}|$.
\end{lemma}
\begin{proof}
This statement is proved in~\cite{ZV02}. For the sake of completeness
we present here another proof which is also much simpler than the
original one. The condition that $M^{(j)}$ has full rank is equivalent
to the following condition: for every (non-zero) polynomial
\begin{equation}
F(x_1,\ldots,x_j) = \sum_{\alpha\in\mathcal{I}^{(j)}} c_\alpha
m_\alpha
\end{equation}
there is at least one $v\in\mathcal{V}^{(j)}$ such that $F(v)\neq 0$.

Let $b:=\min\{d,q-1\}$ denote the upper bound on the local
degrees.  For $j=1$, we have $\mathcal{I}^{(1)}=\{(1),\ldots,(b)\}$
and the corresponding set of monomials is
$\{x_1,x_1^2,\ldots,x_1^b\}$. We can choose
$\mathcal{V}^{(1)}:=\{v_1,v_2,\ldots,v_b\}$ to be a set containing $b$
different non-zero elements of $\F_q$.  Then the matrix
\[
M^{(1)}=\left(
\begin{array}{ccccc}
v_1^1  & v_1^2  & \cdots & v_1^b \\
v_2^1  & v_2^2  & \cdots & v_2^b \\
\vdots & \vdots & \ddots & \vdots \\
v_b^1  & v_b^2  & \cdots & v_b^b 
\end{array}
\right)
\]
has the full rank $|\mathcal{I}^{(1)}|=b$.  Observe that
 we obtain a (square) Vandermonde matrix by
multiplying $M^{(1)}$ with
$\mathrm{diag}(v_1^{-1},v_2^{-1},\cdots,v_b^{-1})$ from the left.  We
choose $v_1$ to be equal to $1$.

Assume that we have already determined a suitable
$\mathcal{V}^{(j-1)}$ for some $j\ge 2$.  We show how to obtain
$\mathcal{V}^{(j)}$ using $\mathcal{V}^{(j-1)}$.  

\begin{enumerate}
\item Set $\mathcal{V}^{(j)} \leftarrow \{(1,\ldots,1)\}\subseteq\F_q^j$
\item Set 
\begin{equation}
L^{(j)} \leftarrow \Big[ \, m_\alpha(v) \, \Big]_{v\in\mathcal{V}^{(j)},\,\alpha\in
  \mathcal{I}^{(j)}}
\end{equation}
\item \texttt{REPEAT}
\item \hspace{0.5cm} Determine a (non-trivial) vector
  $c=(c_\alpha)\in\F_q^{|\mathcal{I}^{(j)}|}$ in the kernel of
  $L^{(j)}$
\item \hspace{0.5cm} Set
\[
G(x_1,\ldots,x_j) \leftarrow \sum_{\alpha\in \mathcal{I}^{(j)}}
c_\alpha m_\alpha
\]
\item \hspace{0.5cm} Determine a vector $u\in\F_q^j$ such that $G(u)\neq
  0$
\item \hspace{0.5cm} Set $\mathcal{V}^{(j)} \leftarrow \mathcal{V}^{(j)} \cup \{u\}$
\item \hspace{0.5cm} Add the row vector $\big(m_\alpha(u)\big)_{\alpha\in\mathcal{I}^{(j)}}$ at the bottom of $L^{(j)}$
\item \texttt{UNTIL} the rank of $L^{(j)}$ is maximal
\end{enumerate}

We now explain how the different steps can be implemented efficiently
and why the algorithm produces a valid $\mathcal{V}^{(j)}$.

We can compute a non-trivial vector $c$ in the kernel of $L^{(j)}$
in step 4 with Gaussian elimination.  To find a $u$ with $G(u)\neq
0$ in step 5, we write $G$ in the form
\[
G(x_1,\ldots,x_j) = \sum_{i=1}^b F_i(x_1,\ldots,x_{j-1}) \cdot
x_j^i\in \F_q[x_1,\ldots,x_{j-1}][x_j]\,.
\]
At least one of the polynomials $F_i$ is non-zero because $G$ is
non-zero.  Set $F$ to be the non-zero $F_i$ with the smallest $i$.  We
can write $F$ as
\[
F(x_1,\ldots,x_{j-1}) = \sum_{\beta \in \mathcal{I}^{(j-1)}} d_{\beta}
m_{\beta} \in \F_q[x_1,\ldots,x_{j-1}]
\]
with appropriate coefficients $d_{\beta} \in \F_q$.  There exists a
vector $v=(v_1,\ldots,v_{j-1})\in\mathcal{V}^{(j-1)}$ with $F(v)\neq
0$.  This is because otherwise we would have a non-trivial linear dependency of
the rows of $L^{(j-1)}$ corresponding to the elements in
$\mathcal{V}^{(j-1)}$.  Hence, the polynomial
$P(x):=G(v_1,\ldots,v_{j-1},x)$ is a non-zero univariate polynomial
that can be written as a linear combination of monomials $m_\gamma$
with $\gamma\in\mathcal{I}^{(1)}$.  The element $w\in\F_q$ with
$P(w)\neq 0$ can be found among the elements of $\mathcal{V}^{(1)}$.
We obtain the desired vector $u$ by setting it equal to
$(v_1,\ldots,v_{j-1},w)$.

By adding the new row vector to $L^{(j)}$ in step 8 we achieve that
the vector $c$ is no longer in the kernel of the new augmented matrix.
Hence, we reduced the dimension of the kernel of the linear map
defined by this matrix.  In other words, we have increased its rank by
exactly $1$.  This shows that the algorithm terminates.
\end{proof}
We are now ready to describe the improved reduction.

\begin{lemma}
\label{lemma:vand}
Let $\mathcal{V}^{(n)}$ be as in Lemma~\ref{lemma:reduce}.  Then the
coefficients of the hidden polynomial of Eq.~(\ref{eq hpfgp}) can be
determined by solving the univariate \HPFGP{} for the polynomials
$Q(v_1 x, v_2 x,\ldots, v_n x)$ for all $v\in\mathcal{V}^{(n)}$.
\end{lemma}

\begin{proof}
The unknown polynomials can be expressed as
\[
Q(x_1,\ldots,x_n) = \sum_{\ell=1}^{d} Q_\ell(x_1,\ldots,x_n)
\,,
\]
where $Q_\ell$ denotes the homogeneous part of total degree $\ell$.

For each $v=(v_1,\ldots,v_n)\in\F_q^n$, the substitution $x_i \mapsto
v_i x$ in the hidden multivariate polynomial $Q$ leads to the
univariate polynomial
\[
P_v(x) := Q(v_1 x,\ldots,v_n x) = \sum_{\ell=1}^{d} Q_\ell(v) x^\ell\,.
\]

We determine the coefficients $Q_\ell(v)$ of $P_v(x)$ 
by using the
quantum algorithm for the univariate case.  Let
$z=[q_\alpha]_{\alpha\in \mathcal{I}^{(n)}}^T$ be the column vector
whose entries are the unknown coefficients we seek to learn.  Let
$y=[Q_1(v)+\ldots+Q_d(v)]_{v\in\mathcal{V}^{(n)}}$ be the column
vector whose entries are the sum of evaluations of the homogeneous
part $Q_\ell$ at the points $v\in\mathcal{V}^{(n)}$.  We have $M^{(n)}
z = y$.  Hence, we can recover $y$ since the generalized Vandermonde
matrix $M^{(n)}$ has full rank.
\end{proof}

Theorem~\ref{theorem:new-reduction} follows directly from
Lemma~\ref{lemma:vand}. Note that in the course of the above
reduction, we learn $d|\mathcal{I}^{(n)}| \log_2(q)$ bits by solving
$|\mathcal{I}^{(n)}|$ instances of the univariate case (each instance
yielding exactly $d$ coefficients in $\F_q$).  The absolute lower
bound is given by $|\mathcal{I}^{(n)}|\log_2(q)$, which corresponds to
the number of bits necessary to specify all coefficients of the hidden
polynomial $Q$.  This discussion shows that our method is optimal up
to the factor $d$.

Theorem~\ref{theorem:new-reduction} also gives an instance
of the HSSP which is solvable in quantum polynomial time,
although no small strong bases exist (and therefore 
the reduction scheme of Section~\ref{section:red} does not
work directly). 
Let $A$ and  $B$ be the additive group of $\F_q^n$ and
$\F_q$, respectively, and
and let $K$ be $\F_q^{(d)}[x_1,\ldots,x_n]$, 
the set of polynomials in $n$ variables of total degree at most
$d$. Then
Proposition~\ref{proposition:hpfgp-hsop} translates to the following
statement.
\begin{proposition}\label{proposition:idontknow2}
Let $f : \F_q^n \times \F_q\rightarrow S$ be a function. Then $f$
hides for $\HPFGP(\F_q, n, d)$ the polynomial $Q$ if and only if for
$\HSOP(\Fg(\F_q^{(d)}[x_1,\ldots,x_n]), \F_q^n \times \F_q, \circ, {\cal
  SC})$ it hides the standard complement $A_Q$ by symmetries.
\end{proposition}
Therefore, by Theorem~\ref{theorem:new-reduction}, for constant $d$,
$\HSOP(\Fg(\F_q^{(d)}[x_1,\ldots,x_n]), \F_q^n \times \F_q, \circ, {\cal
  SC})$ can be solved in quantum polynomial time. On the other hand,
as multivariate polynomials have many zeros, there are no bases of
polynomial size for the action $\circ$ for $n\geq 2$, see
Lemma~\ref{lemma:sc-strong-base}.

\section*{Acknowledgments}
The authors are grateful to the anonymous referees for their helpful
remarks and suggestions on a previous version of the paper. Most of
this work was conducted when T.D., G.I. and M.S. were at the Centre
for Quantum Technologies (CQT) in Singapore, and partially funded by
the Singapore Ministry of Education and the National Research
Foundation. Research partially supported by the European Commission
IST STREP project Quantum Computer Science (QCS) 25596, by the French
ANR program under contract ANR-08-EMER-012 (QRAC), by the French MAEE
STIC-Asie program FQIC, and by the Hungarian Scientific Research Fund
(OTKA). P.W. gratefully acknowledges the support from the NSF grant
CCF-0726771 and the NSF CAREER Award CCF-0746600.

\end{document}